\documentclass[12pt]{article}
\usepackage{amsfonts,amsthm,amssymb}
\usepackage[fleqn]{amsmath}
\usepackage{hyperref}
\usepackage{comment}
\newcommand{\be}{\begin{eqnarray}}
\newcommand{\ee}{\end{eqnarray}}
\newcommand{\bez}{\begin{eqnarray*}}
\newcommand{\eez}{\end{eqnarray*}}

\newcommand{\cM}{\mathcal{M}}

\newcommand{\bsy}{\boldsymbol}

\newcommand{\pa}{\partial}
\newcommand{\imag}{\mathrm{i}\,}
\newcommand{\ins}{\leavevmode 
    \vbox{\kern.2em \hrule width1.2ex height0.1ex} 
    \hbox{\vrule width0.1ex height1.2ex depth0.ex \kern.1em} \,}
\theoremstyle{definition}
\newtheorem{theorem}{Theorem}[section]
\newtheorem{proposition}[theorem]{Proposition}
\newtheorem{corollary}[theorem]{Corollary}
\newtheorem{remark}[theorem]{Remark}
\newtheorem{example}[theorem]{Example}

\numberwithin{theorem}{section}
\numberwithin{equation}{section}

\allowdisplaybreaks

\begin{document}

\title{\textbf{Fr\"olicher-Nijenhuis geometry \\ and integrable matrix PDE systems}}

\author{\sc{Folkert M\"uller-Hoissen} \\
         \small  Institute for Theoretical Physics,  University of G\"ottingen \\
         \small  Friedrich-Hund-Platz 1, 37077 G\"ottingen, Germany \\
         \small  folkert.mueller-hoissen@phys.uni-goettingen.de  } 

\date{ }

\maketitle

\begin{abstract}
Given two tensor fields of type $(1,1)$ on a smooth $n$-dimensional manifold $\cM$, such that all their 
Fr\"olicher-Nijenhuis brackets vanish, the algebra of differential forms on $\cM$ becomes a bi-differential 
graded algebra. As a consequence, there are partial differential equation (PDE) systems 
associated with it, which arise as the integrability condition of a system of linear equations and 
possess a binary Darboux transformation to generate exact solutions. 
We recover chiral models and potential forms of the self-dual Yang-Mills, as well as corresponding 
generalizations to higher than four dimensions, and obtain new integrable non-autonomous 
nonlinear matrix PDEs and corresponding systems. 
\end{abstract}

\section{Introduction}
The (de Rham) algebra of differential forms  on a smooth manifold $\cM$ is a $\mathbb{Z}$-graded algebra, 
$\bigwedge(\cM) = \bigoplus_{k \geq 0} \bigwedge^k(\cM)$, where $\bigwedge^k(\cM)$ is the space of differential 
forms of degree $k$. Supplied with 
the exterior derivative $d$, it becomes the prime 
example of a \emph{differential graded algebra}, since $d$ is a graded derivation of degree one, 
\bez
     d :  \mbox{$\bigwedge^k(\cM)$} \rightarrow \mbox{$\bigwedge^{k+1}(\cM)$} \, , \quad k=0,1, \ldots \, , \qquad
     d (\alpha \wedge \beta) = (d \alpha) \wedge \beta + (-1)^{\mathrm{deg}(\alpha)} \alpha \wedge d \beta \, ,
\eez
 and satisfies 
\bez
               d^2 = 0 \, .
\eez
This is the basic structure on which, in particular, gauge theory is built, which is of exceptional importance 
in mathematics and physics.

But there are actually many graded derivations of degree one on $\bigwedge(\cM)$ with the latter property. 
According to Fr\"olicher and Nijenhuis \cite{Froe+Nije56}, these are in correspondence with tensor fields 
$N$ of type $(1,1)$ 
having vanishing Nijenhuis torsion, which is an expression in terms of $N$, forming a tensor field of type $(1,2)$.  
In this way, the exterior derivative corresponds to the identity tensor field, which in fact is a distinguished 
choice for $N$. A different choice leads to a graded derivation $d_N$ of degree one, with 
\bez
              d_N^2 =0 \, , 
\eez
and a different formalism of gauge theory can be built with it. Furthermore, we have 
\bez
              d \, d_N = - d_N \, d \, ,
\eez
hence $(\bigwedge(\cM),d,d_N)$ is a \emph{bi-differential graded algebra}. Moreover, if $N_1$ and $N_2$ 
are two different tensor fields of type $(1,1)$, such that all their Fr\"olicher-Nijenhuis brackets 
(see Section~\ref{sec:FN}) vanish, we have
\bez
        d_{N_1}^2 =0 \, ,  \quad    d_{N_2}^2 =0 \, ,  \quad   d_{N_1} \, d_{N_2} = - d_{N_2} \, d_{N_1} \, ,
\eez 
so that $(\bigwedge(\cM),d_{N_1},d_{N_2})$ is a bi-differential graded algebra.

This makes contact with an approach, which has been called \emph{bi-differential calculus}, 
to -- in the sense\footnote{The key here is that the nonlinear equation arises as the integrability condition 
of a system of linear equations.} 
of Lax \cite{Lax68} -- completely integrable nonlinear partial differential and/or difference equations, 
developed since the turn of the last century \cite{DMH00a,DMH20dc}. Given a bi-differential graded algebra, there 
are certain integrable systems associated with it (see Section~\ref{sec:zerocurv}). 
An exploration of  Fr\"olicher-Nijenhuis bi-differential graded algebras within this approach has not 
been carried out before, and the present work intends to fill this gap. On this route, we recover chiral field equations, 
self-dual Yang-Mills equations and generalizations, in a new way. Moreover, and perhaps more importantly, 
we also obtain new non-autonomous versions of these equations.

We note that a type $(1,1)$ tensor with vanishing Nijenhuis torsion also plays a central role in the theory of 
(in a different sense ``integrable'') hydrodynamic-type systems, see \cite{Tsar91,Lore+Magri05}, for example. 
In bi-Hamiltonian systems, (an operator generalization of) such an object shows up as a recursion operator, 
see, e.g., \cite{GVY08} and references cited there.\footnote{Given a symplectic structure (type $(0,2)$ tensor field) 
and a Poisson structure (type $(2,0)$), their contraction is a type $(1,1)$ tensor field $N$. In this way, 
Fr\"olicher-Nijenhuis geometry makes its appearance in finite-dimensional bi-Hamiltonian systems as a 
``recursion operator".} 
A comprehensive exploration of (Fr\"olicher-) Nijenhuis geometry and its applications has recently been 
undertaken by Bolsinov, Konyaev and Matveev, see in particular \cite{BKM22,BKM24}.

This work requires only a basic knowledge of differential geometry. 
In Section~\ref{sec:FN} we recall some results of Fr\"olicher-Nijenhuis theory, which are used 
in the further sections of this work.  
Section~\ref{sec:zerocurv} exploits a parameter-dependent zero curvature condition in the framework of 
Fr\"olicher-Nijenhuis theory, following general ideas of bi-differential calculus. 
Section~\ref{sec:bDT} specializes a general binary Darboux transformation result \cite{DMH13SIGMA} 
of bi-differential calculus to the latter framework.  
Section~\ref{sec:N1=id} deals with the case, where one of the two type $(1,1)$ tensor fields is the identity. 
Section~\ref{sec:N1,N2} provides examples, where both tensor fields are different from the identity. 
Here we make contact with potential forms of the self-dual Yang-Mills equation and generalizations. 
Finally, Section~\ref{sec:conclusions} contains some additional remarks.

\section{Basics of Fr\"olicher-Nijenhuis theory}
\label{sec:FN}

In the following, $\cM$ denotes a smooth  $n$-dimensional manifold and all tensor fields will also be assumed to be smooth (i.e., $C^\infty$).
Let $N$ be a tensor field of type $(1,1)$ on $\cM$. In local coordinates (and using the summation convention),
\bez
               N = N^\mu{}_\nu \, dx^\nu \otimes \frac{\partial}{\partial x^\mu} \, .
\eez
It determines a $C^\infty(\cM)$-linear map $\mathbb{X} \rightarrow \mathbb{X}$, where $\mathbb{X}$ is the space of vector fields on $\cM$ via
\bez 
           NX = N( X^\nu \frac{\partial}{\partial x^\nu} ) = X^\nu \, N(\frac{\partial}{\partial x^\nu} ) 
                                       = X^\nu \, N^\mu{}_\nu \, \frac{\partial}{\partial x^\mu} \, ,
\eez
and also a $C^\infty(\cM)$-linear map $N^\ast : \bigwedge^1(\cM) \rightarrow \bigwedge^1(\cM)$ 
of the space of differential 1-forms on $\cM$ via 
\bez
              N^\ast \alpha(X) := \alpha(N X) = (NX) \ins \alpha =: (N \ins \alpha)(X)  \, .
\eez
In local coordinates,
\bez
      N^\ast \alpha =  N^\ast ( \alpha_\mu \, dx^\mu ) = \alpha_\mu \, N^\ast dx^\mu  
                                                           = \alpha_\mu \, N^\mu{}_\nu \, dx^\nu \, .
\eez
The first appearance of $\ins$ in the preceding equation is the usual insertion of a vector field in a differential form (interior product), which is a \emph{graded} derivation on the algebra $\bigwedge(\cM)$ of differential forms.   
$N \ins$ is defined as insertion of the vector field part of the type $(1,1)$ tensor $N$ in a differential form, which is a derivation:
\bez
        N \ins (\alpha \wedge \beta) = (N \ins \alpha) \wedge \beta + \alpha \wedge N \ins \beta 
\eez
for any differential forms $\alpha$ and $\beta$. 

Associated with any  type $(1,1)$ tensor field $N$ is a graded derivation $d_N$ of degree one \cite{Froe+Nije56}. 
For any differential form $\omega$, 
\bez
          d_N \omega = N \ins d \omega - d (N \ins \omega) \, ,
\eez
where $d$ is the exterior derivative.  $d_N$ has the property 
\bez
         d \, d_N = - d_N \, d   \, .
\eez
If $N$ is the identity, so that $N^\mu{}_\nu = \delta^\mu_\nu$, then $d_N = d$. 

For a function $f$, we have 
\bez
            d_N f = N \ins df = N^\ast df \, . 
\eez
For a 1-form $\alpha = \alpha_\mu \, dx^\mu$,
\bez
           d_N \alpha = \big( N^\kappa{}_{[\mu} \, \alpha_{\nu],\kappa} + \alpha_\kappa \, N^\kappa{}_{[\mu,\nu]} \big) 
                                     \, dx^\mu \wedge dx^\nu \, ,
\eez
where a comma, followed by a greek index, indicates a partial derivative with respect to the corresponding 
coordinate, and square brackets mean antisymmetrization.

If  $N_1,N_2$ are two tensor fields of type $(1,1)$, we set
\bez
            N_1 N_2 := N_1^\mu{}_\lambda \,   N_2^\lambda{}_\nu \, dx^\nu \otimes \frac{\partial}{\partial x^\mu} \, .
\eez
For 1-forms, we have $(N_1 N_2) \ins \alpha = N_2 \ins (N_1 \ins \alpha)$, i.e., $(N_1 N_2)^\ast \alpha =  N_2^\ast  N_1^\ast \alpha$.

For  two type $(1,1)$ tensor fields, the following holds,
\be
       [d_{N_1} , d_{N_2}] = d_{ [N_1,N_2]_{\mathrm{FN}}} \, ,    \label{d_N1,d_N2_commutator}
\ee
with the \emph{Fr\"olicher-Nijenhuis bracket} 
\bez
    [N_1,N_2]_{\mathrm{FN}}(X,Y) &:=& (N_1 N_2 + N_2 N_1)[X,Y] + [N_1X,N_2Y] +[N_2X,N_1Y] \\
      & &   - N_1 ( [N_2X,Y] + [X,N_2Y] ) - N_2 ( [N_1X,Y] + [X,N_1Y] ) \, ,
\eez
where $X$ and $Y$ are any vector fields on $\cM$. The bracket $[N_1,N_2]_{\mathrm{FN}}$ is bilinear 
and symmetric. The right hand side of (\ref{d_N1,d_N2_commutator}) is a graded derivation of degree 
two, see \cite{Froe+Nije56} for its definition. For this work, it is sufficient to know that it vanishes if 
$[N_1,N_2]_{\mathrm{FN}} = 0$.

In local coordinates, we have
\bez
    [N_1,N_2]_{\mathrm{FN}}(\frac{\pa}{\pa x^\mu},\frac{\pa}{\pa x^\nu}) 
  = 2 \Big( N_1{}^\kappa{}_{[\mu} \, N_2{}^\lambda{}_{\nu],\kappa} + N_2{}^\kappa{}_{[\mu} \, N_1{}^\lambda{}_{\nu],\kappa}
     + N_1{}^\lambda{}_{\kappa} \, N_2{}^\kappa{}_{[\mu,\nu]} + N_2{}^\lambda{}_{\kappa} \, N_1{}^\kappa{}_{[\mu,\nu]}
     \Big) \, \frac{\pa}{\pa x^\lambda} \, .
\eez

The \emph{Nijenhuis torsion} of a type $(1,1)$ tensor field $N$ is the type $(1,2)$ tensor field
\bez
     T(N)(X,Y) := \frac{1}{2} [N,N]_{\mathrm{FN}}(X,Y) = [NX,NY] + N \Big( N [X,Y] - [NX,Y] - [X,NY] \Big)  \, .
\eez

\begin{proposition}[\cite{Froe+Nije56}]
\label{prop:d_N^2=0}
$d_N^2 =0$ if and only if $T(N)=0$. \hfill $\Box$
\end{proposition}

\begin{proposition}
\label{rem:d_N1,d_N2_bi-diff}
For two different tensor fields $N_1,N_2$ of type $(1,1)$, satisfying $[N_i,N_j]_{\mathrm{FN}}=0$, $i,j=1,2$,   
the corresponding graded derivations of degree one, $d_{N_1}$ and $d_{N_2}$, satisfy 
\bez
        d_{N_1}^2 = 0 = d_{N_2}^2 \, , \qquad d_{N_1} \, d_{N_2} = - d_{N_2} \, d_{N_1} \, ,
\eez
and thus supply the de Rham algebra of differential forms with the structure 
of a \emph{double} or \emph{bi-differential graded algebra}, also called a \emph{bi-differential calculus} \cite{DMH00a}.\footnote{This is in particular a cochain double complex (cochain bi-complex).}
\end{proposition}
\begin{proof}
This follows immediately from Proposition~\ref{prop:d_N^2=0} and (\ref{d_N1,d_N2_commutator}).
\end{proof}

\section{Zero curvature condition and integrable PDEs}
\label{sec:zerocurv}
Let $N_i$, $i=1,2$, be two tensor fields of type $(1,1)$ on a manifold $\cM$, satisfying $[N_i,N_j]_{\mathrm{FN}}=0$, $i,j=1,2$, 
and $(\bigwedge(\cM), d_{N_1},d_{N_2})$ the corresponding bi-differential graded algebra. 
Let us introduce covariant exterior derivatives
\bez
         D_{N_1} := d_{N_1} - B \, , \qquad D_{N_2} := d_{N_2} - A \, ,
\eez
acting from the left on $(C^\infty(\cM))^m$, with complex $m \times m$ matrices of 1-forms, 
$A$ and $B$.
The latter may be regarded as local gauge potentials (pullback with respect to a local cross section) of a 
connection 1-form on a $GL(m,\mathbb{C})$ principal bundle over $\cM$, with respect to the 
differential graded algebra $(\bigwedge(\cM), d_{N_1})$, respectively $(\bigwedge(\cM), d_{N_2})$. 
In an obvious way, the covariant exterior derivatives extend to $\mathbb{C}$-linear maps 
$(\bigwedge^r(\cM))^m \otimes_{\mathbb{R}} \mathbb{C} \rightarrow (\bigwedge^{r+1}(\cM))^m \otimes_{\mathbb{R}} \mathbb{C}$, for $r=0,1,2,\ldots$, which constitutes 
a cochain bi-complex if the two connections have vanishing curvature and anticommute, i.e., if
\be
          D_{N_1}^2 = 0 \, , \qquad  D_{N_2}^2 = 0  \, , \qquad  D_{N_1} D_{N_2} = - D_{N_2} D_{N_1} \, .
                                \label{bi-complex_cond}
\ee

A zero curvature condition, somewhat generalized, is well-known to underly (completely) integrable partial 
differential equations. In general, in this respect it is effective only if it depends on a (``spectral'') parameter $\lambda$.
In fact, the conditions (\ref{bi-complex_cond}) are equivalent to
\be
          ( D_{N_1} + \lambda \, D_{N_2} )^2 = 0 \qquad \forall \lambda \, .      \label{bi-complex_cond_lambda}
\ee 

Using the definitions of the two covariant exterior derivatives, (\ref{bi-complex_cond}) is seen to be equivalent to
\be
         d_{N_1} B = B \wedge B  \, , \qquad  d_{N_2} A = A \wedge A  \, , \qquad 
         d_{N_1} A +  d_{N_2} B - A \wedge B - B \wedge A = 0 \, .    \label{zero_curv_A,B}
\ee

The gauge potentials $A$ and $B$ are not determined by (\ref{bi-complex_cond}), since they are subject to gauge 
transformations, which can be used to restrict the form of $A$ and $B$ without loss of generality. 
For example, if $d_{N_1} = d$, the exterior derivative, it is well known that $B=0$ can be achieved locally by 
a gauge transformation. For degenerate $N_1$, this is not true. 
Nevertheless, we will impose the condition
\be
            B = 0    \label{B=0}
\ee
in the following, in order to somewhat simplify our explorations. It is important to note that we cannot achieve simultaneously $A=0$. 
By breaking the gauge invariance of (\ref{zero_curv_A,B}) via imposing the condition (\ref{B=0}), it reduces to
\be
        d_{N_1} A = 0  \, , \qquad  d_{N_2} A = A \wedge A  \, .  \label{A-eqs}
\ee

Next we consider two reductions of (\ref{A-eqs}). Imposing 
\be
           A = d_{N_1} \phi \, ,    \label{A=dN1phi}
\ee
with an $m \times m$ matrix of scalars, (\ref{A-eqs}) takes the form
\be
        d_{N_2} d_{N_1} \phi = d_{N_1} \phi \wedge d_{N_1} \phi  \, ,  \label{phi_eq}
\ee
In arbitrary coordinates $x^\mu$, $\mu=1,\ldots,n$, this reads   
\be
        N_2{}^\kappa{}_{[\mu} \, ( N_1{}^\lambda{}_{\nu]} \, \phi_{,\lambda} )_{,\kappa} 
             + N_1{}^\kappa{}_\lambda \, N_2{}^\lambda{}_{[\mu,\nu]} \, \phi_{,\kappa}
       = \frac{1}{2} N_1{}^\kappa{}_\mu \, N_1{}^\lambda{}_\nu \, [\phi_{,\kappa} \, , \phi_{,\lambda}] \, .
                      \label{phi_N1_N2_eq}
\ee

Alternatively, setting 
\be
           A = (d_{N_2} g) \, g^{-1} \, ,  \label{A=(dN2g)g^-1}
\ee
with an invertible $m \times m$ matrix of scalars, (\ref{A-eqs}) results in 
\be
              d_{N_1} ( (d_{N_2} g) \, g^{-1}) = 0  \, .   \label{g_eq}
\ee
In coordinates $x^\mu$, this reads
\be
        N_1{}^\kappa{}_{[\mu} \, \Big( N_2{}^\lambda{}_{\nu]} \, g_{,\lambda} \, g^{-1} \Big)_{,\kappa} 
       +  N_2{}^\kappa{}_\lambda \, N_1{}^\lambda{}_{[\mu,\nu]} \, g_{,\kappa}  \, g^{-1}  = 0 \, . 
                               \label{g_N1_N2_eq}
\ee

The two equations (\ref{phi_eq}) and (\ref{g_eq}) are related by the Miura-type transformation
\bez
          (d_{N_2} g) \, g^{-1} = d_{N_1} \phi \, ,
\eez
which has both, (\ref{phi_eq}) and (\ref{g_eq}), as integrability conditions. In local coordinates,
\bez
          N_2{}^\mu{}_\nu \, g_{,\mu} \, g^{-1} =  N_1{}^\mu{}_\nu \, \phi_{,\mu} \, .
\eez

\begin{remark}
\label{rem:symm}
(\ref{phi_eq}) is invariant under $\phi \mapsto \phi + \phi_0$, where $\phi_0$ is any $d_{N_1}$-constant $m \times m$ 
matrix of $0$-forms. (\ref{g_eq}) is invariant under $g \mapsto g \, g_0$, with any invertible $d_{N_2}$-constant 
$m \times m$ matrix $g_0$ of $0$-forms. \hfill $\Box$
\end{remark}

\begin{remark}
An equivalent expression for the bi-complex conditions (\ref{bi-complex_cond_lambda}) with $B=0$ is
\be
        [ N_2{}^\kappa{}_\mu \, \partial_\kappa - A_\mu + \lambda \, N_1{}^\kappa{}_\mu \, \partial_\kappa \, , \, 
         N_2{}^\kappa{}_\nu \, \partial_\kappa - A_\nu + \lambda \, N_1{}^\kappa{}_\nu \, \partial_\kappa ] = 0 
       \qquad  \mu,\nu = 1,\ldots, n \, .     \label{0curv_system}
\ee
\hfill $\Box$
\end{remark}

\subsection{A chain of conservation laws}
Let the bi-complex conditions (\ref{bi-complex_cond}) hold. 
We define a \emph{conservation law} to be given by an element 
$\rho \in (\bigwedge^1(\cM))^m \otimes_{\mathbb{R}} \mathbb{C}$ such that
\bez
          D_{N_1} \rho = 0 \qquad \mbox{and} \qquad D_{N_2}  \rho = 0 \, .
\eez
Let $\chi_0 \in (C^\infty(\cM))^m \otimes_{\mathbb{R}} \mathbb{C}$ satisfy $D_{N_1} D_{N_2} \chi_0 = 0$. Then 
\bez
                       \rho_1 := D_{N_2} \chi_0
\eez
is $D_{N_2}$-closed, but also $D_{N_1}$-closed since
\bez
         D_{N_1} \rho_1 = D_{N_1} D_{N_2} \chi_0 = 0 \, .
\eez
If this implies the existence of $\chi_1 \in (C^\infty(\cM))^m \otimes_{\mathbb{R}} \mathbb{C}$ such that 
$\rho_1 = D_{N_1} \chi_1$, 
then 
\bez
           \rho_2 := D_{N_2}  \chi_1
\eez 
is $D_{N_2} $-closed and we also have
\bez
          D_{N_1}  \rho_2 = D_{N_1}  D_{N_2}  \chi_1 = - D_{N_2}  D_{N_1}  \chi_1 = - D_{N_2}  \rho_1 = 0 \, ,
\eez
and so forth, leading to a chain of conservation laws
\bez
       \rho_{k+1} = D_{N_1}  \chi_{k+1} = D_{N_2}  \chi_k  \quad \qquad k=0,1,2, \ldots \, .
\eez 
In terms of 
\bez
          \chi := \sum_{k=0}^\infty \lambda^k \, \chi_k   \, ,
\eez
with a parameter $\lambda$, we have the linear equation
\bez
        D_{N_1}  \chi = \lambda \, D_{N_2} \chi  \, .    
\eez
If the bi-complex conditions (\ref{bi-complex_cond}) are equivalent to some PDE, 
then the above construction typically yields an 
infinite chain of conservation laws for this PDE. 
Also see \cite{BIZZ79,PSW79,Magr+Moro84,DMH96int,DMH97int,DMH00a}.

\section{Binary Darboux transformations}
\label{sec:bDT}

Bi-differential calculus \cite{DMH00a,DMH20dc} has been developed more generally for associative graded 
algebras. We recall a useful result \cite{DMH13SIGMA}, here restricted to the case of a Fr\"olicher-Nijenhuis 
bi-differential graded algebra.

\begin{theorem}
\label{thm:bDT}
 Let $N_1$ and $N_2$ be two tensor fields of type $(1,1)$ on an $n$-dimensional manifold $\cM$, with 
$[N_i, N_j]_{\mathrm{FN}}=0$, $i,j=1,2$, so that $(\bigwedge(\cM), d_{N_1},d_{N_2})$ is a bi-differential graded algebra.
Let $r \times r$ matrices $Q, P$ of $0$-forms and $r \times r$ matrices $\bsy{\kappa},\bsy{\lambda}$ of $1$-forms  satisfy  
\be
  &&   d_{N_2} P - [\bsy{\kappa} , P] = P \,  d_{N_1} P  \, , \hspace{1.15cm}
       d_{N_2} \bsy{\kappa} - \bsy{\kappa} \wedge \bsy{\kappa} = P \,  d_{N_1} \bsy{\kappa}  \, , \nonumber \\ 
  &&   d_{N_2} Q + [\bsy{\lambda} , Q] = ( d_{N_1} Q) \, Q \, , \qquad  
       d_{N_2} \bsy{\lambda} + \bsy{\lambda} \wedge \bsy{\lambda} = ( d_{N_1} \bsy{\lambda}) \, Q  \, .   
                    \label{Delta,lambda,Gamma,kappa_eqs}
\ee
Furthermore, let $m \times r$, respectively $r \times m$, matrices $\theta, \eta$ of $0$-forms be solutions of the 
homogeneous linear systems
\be
    d_{N_2} \theta &=& A \, \theta + (d_{N_1} \theta) \, Q + \theta \, \bsy{\lambda}  \, ,  \label{lin_sys1} \\
    d_{N_2} \eta &=& - \eta \, A + P \, d_{N_1} \eta +\bsy{\kappa} \, \eta \, ,  \label{lin_sys2}
\ee
where $A$ is an $m \times m$ matrix of 1-forms solving (\ref{A-eqs}). 
If $\Omega$ is an $r \times r$ matrix solution of the (compatible) inhomogeneous linear equations
\be
    && P \, \Omega - \Omega \, Q = \eta \, \theta  \, , \label{Sylv}  \\
    && d_{N_2} \Omega = (d_{N_1} \Omega) \, Q - (d_{N_1} P) \, \Omega + (d_{N_1} \eta) \, \theta 
            + \bsy{\kappa} \, \Omega + \Omega \, \bsy{\lambda} \, ,
            \label{d_NOmega}
\ee
then  
\be
                    A' := A - d_{N_1} (\theta \, \Omega^{-1} \eta)   \label{new_A}
\ee
also solves (\ref{A-eqs}) on any open set of $\cM$, where $\Omega$ is invertible.  \hfill $\Box$
\end{theorem}

\begin{remark}
The 1-forms $\bsy{\kappa}, \bsy{\lambda}$ represent a certain freedom in the formalism, which turned 
out to be helpful, since simplifying, in some examples of integrable systems treated in the 
bi-differential calculus framework, see \cite{CDMH16}, for example. 
Setting $\bsy{\kappa} = \bsy{\lambda} =0$, (\ref{Delta,lambda,Gamma,kappa_eqs}) reduces to
\bez
       d_{N_2} P = P \,  d_{N_1} P  \, , \qquad
      d_{N_2} Q = ( d_{N_1} Q) \, Q \, , 
\eez
which are generalizations of matrix versions of the Riemann (or Hopf, or inviscid Burgers) equation.
Choosing instead $\bsy{\kappa} = d_{N_1} P$, $\bsy{\lambda} = d_{N_1} Q$, (\ref{Delta,lambda,Gamma,kappa_eqs}) 
reduces to 
\bez
         d_{N_2} P = (d_{N_1} P) \, P  \, , \qquad
         d_{N_2} Q = Q \, d_{N_1} Q \, .
\eez
There are many more possibilities, of course. The choice modifies the linear systems (\ref{lin_sys1}),  
(\ref{lin_sys2}), and it will therefore be useful for 
comparison with linear systems found in different ways. \hfill $\Box$
\end{remark}

\begin{corollary}
Let all assumptions of Theorem~\ref{thm:bDT} hold with (\ref{A=dN1phi}). Then,  starting with a ``seed" solution 
$\phi$ of (\ref{phi_eq}),  
\be
               \phi' = \phi - \theta \, \Omega^{-1} \, \eta   \label{phi_new}
\ee
is a new solution of (\ref{phi_eq}).
\end{corollary}
\begin{proof}
(\ref{new_A}) takes the form $d_{N_1}(\phi' - \phi + \theta \, \Omega^{-1} \eta)=0$, hence 
$\phi' = \phi - \theta \, \Omega^{-1} \, \eta + \phi_0$, where $d_{N_1}\phi_0 = 0$. We can achieve $\phi_0=0$ 
by a symmetry of (\ref{phi_eq}), see Remark~\ref{rem:symm}.
\end{proof}

\begin{corollary}
Let the assumptions of Theorem~\ref{thm:bDT} hold with (\ref{A=(dN2g)g^-1}) and let $P$ be invertible. 
Then, given a solution $g$ of  (\ref{g_eq}),  
\be
            g' = (I - \theta \,  \Omega^{-1} P^{-1} \eta) \, g   \label{g_new}
\ee
is a new solution of (\ref{g_eq}). 
\end{corollary}
\begin{proof} 
Using (\ref{lin_sys1}), (\ref{lin_sys2}), the equation for $d_{N_2} P$ in (\ref{Delta,lambda,Gamma,kappa_eqs}), and (\ref{Sylv}), with $g'$ given by (\ref{g_new}) we obtain
\bez
         (d_{N_2} g') \, {g'}^{-1} = (d_{N_2} g) \, g^{-1} - d_{N_1}(\theta \,  \Omega^{-1} \eta) 
       = A - d_{N_1}(\theta \,  \Omega^{-1} \eta)  = A'  \, ,
\eez 
using (\ref{new_A}) in the last step.
\end{proof}

According to the preceding results, the two equations (\ref{phi_eq}) and (\ref{g_eq}) are integrable in the 
sense that they are the integrability condition of both linear systems (\ref{lin_sys1}) and (\ref{lin_sys2}). 
Moreover, they possess a binary Darboux transformation to generate new from given solutions.

\begin{remark}
(\ref{Sylv}) is a matrix \emph{Sylvester equation}. Assuming that $Q$ and $P$ have no eigenvalue 
in common, the Sylvester equation has a unique solution.  If this ``spectrum condition" holds, then (\ref{d_NOmega}) is 
automatically satisfied. If the spectrum condition is violated, the Sylvester equation leads to constraints for the 
solution of the linear systems and does not completely determine $\Omega$. But this is then taken 
into account by (\ref{d_NOmega}).  
See \cite{Manas96} for the important example of dark solitons of the defocusing NLS equation. 
\hfill $\Box$
\end{remark}

In local coordinates, (\ref{Delta,lambda,Gamma,kappa_eqs}) has the form
\bez
  &&  N_2{}^\nu{}_\mu \, P_{,\nu} - [ \bsy{\kappa}_\mu,P] = N_1{}^\nu{}_\mu \, P P_{,\nu} \, , \quad
          N_2{}^\nu{}_\mu \, Q_{,\nu} + [ \bsy{\lambda}_\mu , Q ] = N_1{}^\nu{}_\mu \, Q_{,\nu}  \, Q \, , \\
  &&  N_2{}^\kappa{}_{[\mu} \, \bsy{\kappa}_{\nu],\kappa} + N_2{}^\kappa{}_{[\mu,\nu]}  \, \bsy{\kappa}_\kappa
           - \frac{1}{2} [ \bsy{\kappa}_\mu , \bsy{\kappa}_\nu ] 
     = P \, ( N_1{}^\kappa{}_{[\mu} \, \bsy{\kappa}_{\nu],\kappa} + N_1{}^\kappa{}_{[\mu,\nu]} \, \bsy{\kappa}_\kappa ) \, , \\
   &&  N_2{}^\kappa{}_{[\mu} \, \bsy{\lambda}_{\nu],\kappa} + N_2{}^\kappa{}_{[\mu,\nu]} \, \bsy{\lambda}_\kappa 
           + \frac{1}{2} [ \bsy{\lambda}_\mu , \bsy{\lambda}_\nu ] 
      = ( N_1{}^\kappa{}_{[\mu} \, \bsy{\lambda}_{\nu],\kappa} + N_1{}^\kappa{}_{[\mu,\nu]} \, \bsy{\lambda}_\kappa ) \, Q \, , 
\eez
where we wrote $\bsy{\kappa} = \bsy{\kappa}_\mu \, dx^\mu$ and $\bsy{\lambda} = \bsy{\lambda}_\mu \, dx^\mu$.
The linear systems (\ref{lin_sys1}) and (\ref{lin_sys2}) are
\bez
             \theta_{,\nu} \, ( N_2{}^\nu{}_\mu \, I - N_1{}^\nu{}_\mu \, Q ) 
          = A_\mu \, \theta  + \theta \bsy{\lambda}_\mu \, , \qquad
            ( N_2{}^\nu{}_\mu \, I -  N_1{}^\nu{}_\mu \, P ) \, \eta_{,\nu} 
          = - \eta \, A_\mu + \bsy{\kappa}_\mu \, \eta  \, ,
\eez
where $I$ is the $r \times r$ identity matrix and
\bez
                   A_\mu = N_1{}^\kappa{}_\mu \, \phi_{,\kappa} \, ,
\eez
respectively
\bez
                 A_\mu = N_2{}^\nu{}_\mu \, g_{,\nu} g^{-1} \, .
\eez
Finally, (\ref{d_NOmega}) takes the form
\bez
       N_2{}^\nu{}_\mu \, \Omega_{,\nu} = N_1{}^\nu{}_\mu \, \Omega_{,\nu} \, Q - N_1{}^\nu{}_\mu \, P_{,\nu} \, \Omega
              + N_1{}^\nu{}_\mu \, \eta_{,\nu} \, \theta + \bsy{\kappa}_\mu \Omega + \Omega \bsy{\lambda}_\mu \, .
\eez

\section{The case where one of the two tensor fields is the identity}
\label{sec:N1=id}

We choose $N_1$ as the identity on $\mathbb{X}$, so that $d_{N_1} = d$, the exterior derivative. We  
write $N_2 = N$ and assume that it has vanishing Nijenhuis torsion, i.e.,
\bez
   N^\kappa{}_{[\mu} \, N^\lambda{}_{\nu],\kappa} + N^\kappa{}_{[\mu,\nu]}  \, N^\lambda{}_{\kappa} = 0 \, .
\eez
(\ref{phi_N1_N2_eq}) and (\ref{g_N1_N2_eq}) now take the form
\be
     N^\kappa{}_{[\mu} \, \phi_{,\nu] \kappa} + N^\kappa{}_{[\mu,\nu]} \, \phi_{,\kappa}
     = \phi_{,[\mu} \, \phi_{,\nu]} = \frac{1}{2} [\phi_{,\mu} , \phi_{,\nu} ] \, ,    \label{phi_N_eq}
\ee
respectively
\be
      (g_{,\lambda}  \, g^{-1})_{,[\mu}  \, N^\lambda{}_{\nu]} - g_{,\lambda}  \, g^{-1} \, N^\lambda{}_{[\mu,\nu]} = 0 \, . 
                               \label{g_N_eq}
\ee

\begin{remark}
If $N$ is diagonalizable in the sense that there exist (real) coordinates in which 
$N^\mu{}_\nu = f_\nu \,  \delta^\mu_\nu$ 
with distinct functions $f_\nu$, $\nu=1,\ldots,n$, of the coordinates $x^1,\ldots,x^n$,  then vanishing 
Nijenhuis torsion requires that $f_\nu$ only depends on the single coordinate $x^\nu$. 
In this case, we have $d_N \phi = f_\nu \, \phi_{,\nu} \,  dx^\nu$,
which implies the \emph{linear} equation $d \, d_N \phi = 0$, hence (\ref{phi_eq}) is then not an interesting equation. 
Therefore we shall assume that $N$ is not diagonalizable.  \hfill $\Box$
\end{remark}

\begin{example}
If the matrix $(N^\mu{}_\nu)$ is constant in a coordinate system $x^\mu$, its Nijenhuis torsion vanishes. 
(\ref{phi_N_eq}) and (\ref{g_N_eq}) then reduce to
\bez
      N^\kappa{}_\mu  \, \phi_{,\nu \kappa}  -  N^\kappa{}_\nu \, \phi_{,\mu \kappa}  =  [\phi_{,\mu} , \phi_{,\nu} ] 
     \, ,  \qquad 
    (g_{,\lambda}  \, g^{-1})_{,[\mu}  \, N^\lambda{}_{\nu]} = 0   \, ,
\eez
respectively.  \hfill $\Box$
\end{example}

Let us turn to the conditions in the binary Darboux transformation theorem.  
(\ref{Delta,lambda,Gamma,kappa_eqs}) now reads
\bez
  &&  ( N^\nu{}_\mu \, I - \delta^\nu_\mu \, P ) \, P_{,\nu} - [ \bsy{\kappa}_\mu,P] = 0 \, , \quad
          Q_{,\nu} \, (  N^\nu{}_\mu \, I -  \delta^\nu_\mu \, Q ) + [ \bsy{\lambda}_\mu , Q ] = 0 \, , \\
  &&  N^\kappa{}_{[\mu} \, \bsy{\kappa}_{\nu],\kappa} + N^\kappa{}_{[\mu,\nu]}  \, \bsy{\kappa}_\kappa
           - \frac{1}{2} [ \bsy{\kappa}_\mu , \bsy{\kappa}_\nu ] + P \, \bsy{\kappa}_{[\mu,\nu]} = 0 \, , \\
  &&  N^\kappa{}_{[\mu} \, \bsy{\lambda}_{\nu],\kappa} + N^\kappa{}_{[\mu,\nu]} \, \bsy{\lambda}_\kappa 
           + \frac{1}{2} [ \bsy{\lambda}_\mu , \bsy{\lambda}_\nu ] + \bsy{\lambda}_{[\mu,\nu]} \, Q = 0 \, , 
\eez
and the linear systems (\ref{lin_sys1}) and (\ref{lin_sys2}) are reduced to
\bez
   \theta_{,\nu} \, (N^\nu{}_\mu \, I - \delta^\nu_\mu \, Q) = A_\mu \, \theta + \theta \, \bsy{\lambda}_\mu \, , \qquad
   (N^\nu{}_\mu \, I - \delta^\nu_\mu \, P) \, \eta_{,\nu} = - \eta \, A_\mu + \bsy{\kappa}_\mu \, \eta  \, .
\eez
where 
\bez
                   A_\mu = \phi_{,\mu} \, ,
\eez
respectively
\bez
                 A_\mu = N^\nu{}_\mu \, g_{,\nu} g^{-1} \, .
\eez
Finally, (\ref{d_NOmega}) takes the form
\bez
       N^\nu{}_\mu \, \Omega_{,\nu} = \Omega_{,\mu} \, Q - P_{,\mu} \, \Omega
              + \eta_{,\mu} \, \theta + \bsy{\kappa}_\mu \Omega + \Omega \bsy{\lambda}_\mu \, .
\eez
In addition, we still have the Sylvester equation (\ref{Sylv}), of course.

\begin{remark}
Further integrable systems are obtained by choosing $N_2$ as the identity and $N_1$ different from it. 
 \hfill $\Box$
\end{remark}

\subsection{Integrable nonlinear matrix PDEs in two dimensions}
Let $n=2$ and, in some coordinate system,
\bez
          (N^\mu{}_\nu) = \left( \begin{array}{cc} a & -b \\
                                                                          b & a   \end{array} \right)  \, ,
\eez
with (real) functions $a,b$, related by $\mathrm{tr}(N)= N^\mu{}_\mu =2a$ and 
$\mathrm{tr}(N^2) = N^\mu{}_\nu N^\nu{}_\mu = 2(a^2-b^2)$ to scalars. 
Then $T(N)=0$ if and only if $a$ and $b$ satisfy the Cauchy-Riemann (CR) equations
\bez
         a_{,1} = b_{,2} \, , \qquad a_{,2} = - b_{,1} \, .
\eez
Solutions are obtained by decomposing any holomorphic function $h(z)$ of a single complex variable 
\bez
              z=x^1 + \imag  x^2 \, ,
\eez
on an open subset of $\mathbb{C}$, into real part $a$ and imaginary part $b$, i.e.,
\bez
                 h(z) = a(x^1,x^2) + \imag b(x^1,x^2) \, .
\eez
In two dimensions, the components $N^\mu{}_\nu$ of any torsion-free tensor field $N$ with two 
distinct complex (not real) eigenvalues at each point can be transformed to the above form 
(see, e.g., \cite{BKM22}, also for corresponding results in higher dimensions). 
We note that $h$ and its complex conjugate $h^\ast$ are the eigenvalues of $N$.

For a function $f$, we have
\bez
            d_N f = (a \, f_{,1} + b \, f_{,2} ) \, dx^1 +  (a \, f_{,2} - b \, f_{,1} ) \, dx^2 \, .
\eez
By use of the CR equations, (\ref{phi_N_eq}) now takes the form
\be
            b \, \Delta \phi = [ \phi_{,1} , \phi_{,2} ] \, ,   \label{pdcm}
\ee
with the Euclidean Laplacian $\Delta$ in two dimensions, and (\ref{g_N_eq}) becomes
\bez
        \Big( ( a \, g_{,1} + b \, g_{,2} ) \, g^{-1} \Big)_{,2} -  \Big( ( a \, g_{,2} - b \, g_{,1} ) \, g^{-1} \Big)_{,1} = 0 \, ,
\eez
which, by using the CR equations, reduces to
\be
    b \, \Big( (g_{,1} \, g^{-1} )_{,1} +  (g_{,2} \, g^{-1} )_{,2} \Big) = a \, [ g_{,1} \, g^{-1} ,  g_{,2} \, g^{-1} ]   \, . \label{cm}
\ee
The Miura transformation reads
\bez
    (a \, g_{,1} + b \, g_{,2}) \, g^{-1} = \phi_{,1} \, , \qquad
    (a \, g_{,2} - b \, g_{,1} )  \, g^{-1} = \phi_{,2} \, .
\eez

Theorem~\ref{thm:bDT} supplies  (\ref{pdcm}) and also (\ref{cm}) with a binary Darboux transformation. 
After decoupling of systems of PDEs into systems of first order ODEs, we obtain the following result from
Theorem~\ref{thm:bDT}.

\begin{corollary}
\label{cor:cm}
Let $a$ and $b \neq 0$ be real and imaginary part, respectively, of any holomorphic function (of $z=x^1 + \imag x^2$). 
Let $r \times r$ matrices $P,Q,\bsy{\kappa}_1,\bsy{\kappa}_2, \bsy{\lambda}_1, \bsy{\lambda}_2$ be solutions of
\bez
        \big(  (P-a \, I)^2 + b^2 I \big) \, P_{,1} &=& -(P-a I) [\bsy{\kappa}_1,P] - b \, [\bsy{\kappa}_2,P] \, , \\
        \big( (P-a \, I)^2 + b^2 I \big) \, P_{,2} &=&  -(P-a I) [\bsy{\kappa}_2,P] + b \, [\bsy{\kappa}_1,P]\, , \\
       Q_{,1}  \big( (Q -a \, I)^2 + b^2 I \big) &=& [ \bsy{\lambda}_1,Q] (Q-a I) + b \, [ \bsy{\lambda}_2,Q] \, , \\
       Q_{,2}  \big( (Q -a \, I)^2 + b^2 I \big) &=&  [\bsy{\lambda}_2,Q] (Q-a I) - b \, [\bsy{\lambda}_1,Q] \, ,
\eez
and
\bez
   &&  b \, (\bsy{\kappa}_{1,1} + \bsy{\kappa}_{2,2} )
    + (P - a \, I) \, (\bsy{\kappa}_{1,2} - \bsy{\kappa}_{2,1}) - [\bsy{\kappa}_1,\bsy{\kappa}_2] = 0 \, , \\
   && b \, (\bsy{\lambda}_{1,1} + \bsy{\lambda}_{2,2} )
    + (\bsy{\lambda}_{1,2} - \bsy{\lambda}_{2,1}) \, (Q - a \, I) + [\bsy{\lambda}_1,\bsy{\lambda}_2] = 0 \, .
\eez
Furthermore, let $m \times r$, respectively $r \times m$ matrices $\theta$ and $\eta$ be solutions 
of the linear systems
\bez
       \theta_{,1} \big( (Q-a I)^2 + b^2 I \big) &=& 
    - \theta \, \big( b \bsy{\lambda}_2 + \bsy{\lambda}_1 (Q- a I)\big) - A_1 \theta (Q- a I) - b A_2 \theta  \\
      \theta_{,2} \big( (Q-a I)^2 + b^2 I \big) &=& 
        \theta \, \big( b \bsy{\lambda}_1 - \bsy{\lambda}_2 (Q- a I) \big) + b A_1 \theta - A_2 \theta  (Q- a I) \, , \\
    \big( (P-a \, I)^2 + b^2 I \big) \, \eta_{,1} &=& - \big( b \, \bsy{\kappa}_2 + (P-a \, I) \, \bsy{\kappa}_1 \big) \, \eta 
                    + (P-a I) \eta A_1 + b \, \eta A_2    \, , \\
      \big( (P-a \, I)^2 + b^2 I \big) \, \eta_{,2} &=&  \big( b \, \bsy{\kappa}_1 - (P-a \, I) \, \bsy{\kappa}_2 \big) \, \eta 
                  - b \, \eta A_1 + (P - a I) \eta A_2 \, ,
\eez
where either
\bez
            A_1 = \phi_{,1} \, , \qquad A_2 = \phi_{,2} \, ,
\eez
with a solution $\phi$ of (\ref{pdcm}), or
\bez
          A_1 = a \, g_{,1} g^{-1} + b \, g_{,2} \, g^{-1} \, , \qquad
           A_2 = a \, g_{,2} \, g^{-1} - b \, g_{,1} g^{-1} \, ,
\eez
with a solution $g$ of (\ref{cm}). 
Moreoever, let $\Omega$ be an $r \times r$ matrix solution of
\bez
     &&    P \Omega - \Omega Q = \eta \theta \, , \\
     &&  \Omega_{,1} \, \big( (Q-a \, I)^2 + b^2 I \big) \, 
              = \big( P_{,1} \, \Omega  - \eta_{,1} \, \theta -  \bsy{\kappa}_1 \Omega - \Omega  \bsy{\lambda}_1 \big) (Q - a \, I)  
             + b \, \big( P_{,2} \, \Omega - \eta_{,2} \, \theta - \bsy{\kappa}_2 \Omega - \Omega \, \bsy{\lambda}_2 \big)  \, , \\
    && \Omega_{,2} \, \big( (Q-a \, I)^2 + b^2 I \big) 
             =  \big( P_{,2} \, \Omega -   \eta_{,2} \, \theta -  \bsy{\kappa}_2 \Omega - \Omega \, \bsy{\lambda}_2 \big) (Q - a \, I) 
                - b \, \big( P_{,1} \, \Omega  - \eta_{,1} \, \theta  - \bsy{\kappa}_1 \Omega - \Omega \, \bsy{\lambda}_1 \big) .
\eez
Then
\bez
          \phi' = \phi - \theta \, \Omega^{-1} \eta \, , \qquad
           g' = (I - \theta \, \Omega^{-1} P^{-1} \eta) \, g
\eez
also solves (\ref{pdcm}), respectively (\ref{cm}), on any open set of $\mathbb{R}^2$, where $\Omega$ is invertible.  
In the second case, we assume that $P$ is invertible. \hfill $\Box$
\end{corollary}

In terms of the complex variable $z$ and its conjugate $\bar{z}=x^1-\imag x^2$, (\ref{pdcm}) and (\ref{cm}) read
\be
         \mathrm{Im}(h) \,  \phi_{,z \bar{z}} = - \frac{\imag}{2} \, [ \phi_{,z} , \phi_{,  \bar{z}} ] \, ,   \label{pdcm_z}
\ee
respectively
\be
     \mathrm{Im}(h) \, \Big( (g_{,z} \, g^{-1} )_{,\bar{z}} +  (g_{,\bar{z}} \, g^{-1} )_{,z} \Big) 
               = - \imag  \mathrm{Re}(h) \, [ g_z \, g^{-1} ,  g_{,\bar{z}} \, g^{-1} ]   \, .       \label{cm_z}
\ee
These equations are invariant under holomorphic transformations of $z$. 
If $a = \mathrm{Re}(h) =0$, the last equation has the form of the principal chiral model equation, 
which is formally obtained by 
replacing $z$ and $\bar{z}$ with real independent variables. Corollary~\ref{cor:cm} now takes the following 
simpler form, writing $A = A_z \, dz + A_{\bar{z}} \, d \bar{z}$, $\bsy{\kappa} = \bsy{\kappa}_z \, dz 
+ \bsy{\kappa}_{\bar{z}} \, d \bar{z}$ and $\bsy{\lambda} = \bsy{\lambda}_z \, dz 
+ \bsy{\lambda}_{\bar{z}} \, d \bar{z}$.

\begin{theorem}
Let $\mathcal{U} \subset \mathbb{C}$ be symmetric about the real axis and $h(z)$ a holomorphic function 
on $\mathcal{U}$. 
Let $r \times r$ matrices $P,Q, \bsy{\kappa}_z,\bsy{\kappa}_{\bar{z}}, \bsy{\lambda}_z, \bsy{\lambda}_{\bar{z}}$ 
be solutions of
\bez
  &&  \big( P- h I \big) \, P_{,z} = - [\bsy{\kappa}_z , P]  \, , \quad
     \big( P- h^\ast I \big) \, P_{,\bar{z}} = - [\bsy{\kappa}_{\bar{z}} , P]  \, , \\
  &&  Q_{,z} \, \big( Q- h I \big) = [\bsy{\lambda}_z , Q] \, , \hspace{.7cm}
      Q_{,\bar{z}} \, \big( Q- h^\ast I \big) = [\bsy{\lambda}_{\bar{z}} , Q] \, , \\
  &&  \big( P - h^\ast I \big) \, \bsy{\kappa}_{z,\bar{z}} - \big( P - h I \big) \, \bsy{\kappa}_{\bar{z},z} 
           - [  \bsy{\kappa}_z ,  \bsy{\kappa}_{\bar{z}} ] = 0 \, , \\
  &&  \bsy{\lambda}_{z,\bar{z}} \, \big( Q- h^\ast I) - \bsy{\lambda}_{\bar{z},z} \, \big( Q- h I) 
         + [ \bsy{\lambda}_z , \bsy{\lambda}_{\bar{z}} ] = 0 \, ,
\eez
where $P- h I, P- h^\ast I, Q- h I, Q- h^\ast I$ are assumed to be invertible.
Furthermore, let $m \times r$, respectively $r \times m$ matrices $\theta$ and $\eta$ be solutions 
of the linear systems\footnote{These are matrix versions of the standard Lax pair for the chiral field equation \cite{Zakh+Mikh78rel}.}
\bez
     && \theta_{,z} \, \big( Q - h I \big) 
           = -\theta \, \bsy{\lambda}_z - A_z \, \theta \, , \quad
          \theta_{,\bar{z}} \, \big( Q - h^\ast I \big) 
            = - \theta \, \bsy{\lambda}_{\bar{z}} - A_{\bar{z}} \, \theta \, , \\
     && (P - h I) \, \eta_{,z}  = -  \bsy{\kappa}_z \, \eta + \eta \, A_z \, , \quad
            (P - h^\ast I) \, \eta_{,\bar{z}} = - \bsy{\kappa}_{\bar{z}} \, \eta + \eta \, A_{\bar{z}} \, ,
\eez
where either
\bez
        A_z = \phi_{,z} \, , \qquad A_{\bar{z}} = \phi_{, \bar{z}} \, ,
\eez
with a solution $\phi$ of (\ref{pdcm_z}), or 
\bez
         A_z = h \, g_{,z} \, g^{-1} \, , \qquad  A_{\bar{z}} = h^\ast \, g_{,\bar{z}} \, g^{-1}  \, ,
\eez
with a solution $g$ of (\ref{cm_z}). Moreover, let $\Omega$ be an $r \times r$ matrix solution of the compatible 
inhomogeneous linear equations
\bez
     &&    P \Omega - \Omega Q = \eta \theta \, , \\
    && \Omega_{,z} \, (Q - h I) = P_{,z}  - \eta_{,z} \theta  - \bsy{\kappa}_z \Omega - \Omega \bsy{\lambda}_z \, , \qquad
    \Omega_{,\bar{z}} \, (Q - h^\ast I) 
         = P_{,\bar{z}}  - \eta_{,\bar{z}} \theta  - \bsy{\kappa}_{\bar{z}} \Omega - \Omega \bsy{\lambda}_{\bar{z}} \, .
\eez
If $\Omega$ is invertible on $\mathcal{U}$, then
\bez
          \phi' = \phi - \theta \, \Omega^{-1} \eta \, , \qquad
           g' = (I - \theta \, \Omega^{-1} P^{-1} \eta) \, g
\eez
also solves (\ref{pdcm_z}), respectively (\ref{cm_z}) (assuming $P$ invertible in the second case).
\hfill $\Box$
\end{theorem}

\begin{remark}
The general solutions of the linear systems have the form $\theta = \theta_0 U$, 
$\eta = V \eta_0$, with invertible $r \times r$ matrices $U,V$ and constant matrices $\theta_0,\eta_0$. 
If $U$ commutes with $Q$, and $V$ with $P$, which is in particular the case if $A_z=A_{\bar{z}}=0$ (zero seed), 
under the assumption of the spectrum condition the Sylvester equation implies $\Omega = V \Omega_0 U$,
where $P \Omega_0 - \Omega_0 Q = \eta_0 \theta_0$. The expressions for the generated 
solutions $\phi'$, respectively $g'$, now show that $U,V$ drop out, so that the resulting solutions are 
constant.  Non-constant solutions can therefore only be generated if not all of 
the commutators $[A_z,  Q]$,  $[A_{\bar{z}},  Q]$, $[A_z,P]$ and $[A_{\bar{z}},  P]$ vanish. 
In particular, we need non-vanishing seed to generate non-constant solutions. 
For constant $h,A_z, A_{\bar{z}}$, examples are easily worked out. 
\hfill $\Box$
\end{remark}

If the holomorphic function $h$ is not constant, then (\ref{pdcm}) and (\ref{cm}) are, to the best of  
our knowledge, new non-autonomous integrable matrix PDEs in two dimensions. Since the linear and the adjoint 
linear system are non-autonomous in this case, it is more difficult, however, to find exact solutions. Further 
explorations are needed.

\subsection{An example in four dimensions}
What do we get if we choose $N$ block-diagonal with $2 \times 2$ blocks of the form considered in the preceding subsection?
 Let $n=4$ and, in some coordinate system, 
\bez
          (N^\mu{}_\nu) = \left( \begin{array}{cccc} a_1 & -b_1 & 0 & 0 \\
                                                                          b_1 & a_1 & 0 & 0 \\
                                                                          0 & 0 & a_2 & -b_2 \\
                                                                          0 & 0 & b_2 & a_2  \end{array} \right)   \, ,
\eez
so that $N$ is block-diagonal with two blocks of the form considered in the preceding subsection. 
This has vanishing Nijenhuis torsion if the CR equations
\bez
         a_{1,1} = b_{1,2} \, , \quad a_{1,2} = - b_{1,1} \, , \quad 
         a_{2,3} = b_{2,4} \, , \quad a_{2,4} = - b_{2,3} \, ,
\eez
hold, and in addition
\bez
     &&   (a_2 - a_1) \, a_{1,3} + b_2 a_{1,4} + b_1 b_{1,3} = 0 \, , \qquad
              (a_2 - a_1) \, a_{1,4} - b_2 a_{1,3} + b_1 b_{1,4} = 0 \, , \\
     &&  (a_2 - a_1) \, b_{1,3} - b_1 a_{1,3} + b_2 b_{1,4} = 0 \, , \qquad
             (a_2 - a_1) \, b_{1,4} - b_1 a_{1,4} - b_2 b_{1,3} = 0 \, , \\
     &&  (a_2 - a_1) \, a_{2,1} - b_1 a_{2,2} - b_2 b_{2,1} = 0 \, , \qquad
             (a_2 - a_1) \, b_{2,1} - b_1 b_{2,2} + b_2 a_{2,1} = 0 \, , \\
     &&  (a_2 - a_1) \, a_{2,2} + b_1 a_{2,1} - b_2 b_{2,2} = 0 \, , \qquad
             (a_2 - a_1) \, b_{2,2} + b_1 b_{2,1} + b_2 a_{2,2} = 0 \, .
\eez
The latter equations are in particular satisfied if $a_1,b_1$ only depend on $x^1,x^2$, and $a_2,b_2$ only 
on $x^3,x^4$. (\ref{phi_N_eq}) results in two expected equations,
\bez
    b_1 \, ( \phi_{,11} + \phi_{,22}) = [\phi_{,1} , \phi_{,2}] \, , \qquad
    b_2 \, ( \phi_{,33} + \phi_{,44}) = [\phi_{,3} , \phi_{,4}] \, ,
\eez
but in addition we get the system of equations
\bez
       b_1 \phi_{,13} + b_2 \phi_{,24} + (a_2 - a_1) \, \phi_{,23} + [\phi_{,2} , \phi_{,3}] &=& 
         - b_{1,3} \phi_{,1} + a_{1,3} \phi_{,2} - a_{2,2} \phi_{,3} - b_{2,2} \phi_{,4} \, , \\
       b_2 \phi_{,13} + b_1 \phi_{,24} - (a_2 - a_1) \, \phi_{,14} - [\phi_{,1} , \phi_{,4}] &=& 
         - a_{1,4} \phi_{,1} - b_{1,4} \phi_{,2} - b_{2,1} \phi_{,3} + a_{2,1} \phi_{,4} \, , \\
      b_1  \phi_{,14} - b_2 \phi_{,23} + (a_2 - a_1) \, \phi_{,24} + [\phi_{,2} , \phi_{,4}] 
     &=& - b_{1,4} \phi_{,1} + a_{1,4} \phi_{,2} + b_{2,2} \phi_{,3} - a_{2,2} \phi_{,4} \, , \\
    b_2  \phi_{,14} - b_1 \phi_{,23} + (a_2 - a_1) \, \phi_{,13} + [\phi_{,1} , \phi_{,3}] 
     &=& a_{1,3} \phi_{,1} + a_{2,1} \phi_{,3} + b_{1,3} \phi_{,2} + b_{2,1} \phi_{,4} \, ,
\eez
which is quite restrictive. 
This conveys an idea of the complexity we typically meet in higher than two dimensions. 
However, there are more interesting possibilities in the Fr\"olicher-Nijenhuis framework, as we 
will show in the next section.

\section{Examples where both tensor fields are different from the identity}
\label{sec:N1,N2}

\subsection{Self-dual Yang-Mills equation and generalizations}

Let $\cM$ now be a $2k$-dimensional manifold. We assume the existence of local coordinates 
$x^1,\ldots,x^k$, $\bar{x}^1,\ldots,\bar{x}^k$, in which the matrices of $N_1$ and $N_2$ have the form
\bez
          N_1 : \left( \begin{array}{cc} (N^\mu{}_\nu) & 0  \\
                                                                         0 & 0  
                                \end{array} \right)  \, , \quad
          N_2 : \left( \begin{array}{cc}   0 & 0 \\
                                         (N^{\bar{\mu}}{}_\nu)  & 0 
                                 \end{array} \right)  \, ,
\eez
with functions $N^\mu{}_\nu, N^{\bar{\mu}}{}_\nu$ on $\cM$. Here a zero stands for a block of zeros 
and an index $\bar{\mu}$ refers to the coordinate $\bar{x}^\mu$. Hence, for a function $f$, 
\bez
           d_{N_1} f = \sum_{\mu,\nu = 1}^n f_{,\mu} \, N^\mu{}_\nu \, dx^\nu \, , \qquad
           d_{N_2} f = \sum_{\mu,\nu = 1}^n  f_{,\bar{\mu}} \, N^{\bar{\mu}}{}_\nu \, dx^\nu \, .
\eez 
The vanishing of $[N_1,N_1]_{\mathrm{FN}}$  amounts to
\bez
     T( (N^\mu{}_\nu)) := T( N^\lambda{}_\kappa \, dx^\kappa \otimes \frac{\pa}{\pa x^\lambda}) = 0 \quad 
     \mbox{and} \quad N^\lambda{}_\kappa \, N^\kappa{}_{\mu, \bar{\nu}} = 0 \, ,
\eez
where a summation convention (e.g., in the last equation a sum over $\kappa=1,\ldots,k$) is understood, 
and $[N_2,N_2]_{\mathrm{FN}} =0$ becomes
\bez
     N^{\bar{\kappa}}{}_{[\mu} \, N^{\bar{\lambda}}{}_{\nu],\bar{\kappa}} = 0 \, ,
\eez
involving a sum over $\bar{\kappa}=\bar{1},\ldots,\bar{k}$. 
The further condition  $[N_1,N_2]_{\mathrm{FN}} = 0$ requires
\bez
   N^\kappa{}_{[\mu} \, N^{\bar{\lambda}}{}_{\nu], \kappa} + N^{\bar{\lambda}}{}_\kappa \, N^\kappa{}_{[\mu,\nu]} = 0 \, , 
  \quad
   N^{\bar{\kappa}}{}_{[\mu} \, N^\lambda{}_{\nu],\bar{\kappa}} = 0 \, . 
\eez 
In the following, we will assume that $(N^\mu{}_\nu)$ is invertible, so that 
\bez
          N^\mu{}_{\nu,\bar{\kappa}} = 0 \, .
\eez
In the chosen coordinates and with the restrictions imposed on $N_1$ and $N_2$, (\ref{phi_N1_N2_eq}) 
and (\ref{g_N1_N2_eq}) take the form
\be
            N^\lambda{}_\nu \, N^{\bar{\kappa}}{}_\mu \, \phi_{,\lambda \bar{\kappa}} 
          - N^\lambda{}_\mu \, N^{\bar{\kappa}}{}_\nu\, \phi_{,\lambda \bar{\kappa}}   
         =  N^\kappa{}_\mu \, N^\lambda{}_\nu \, [ \phi_{,\kappa} , \phi_{,\lambda}] \, ,   \label{phi_N1_N2_eq_special}
\ee
respectively
\be
      N^\kappa{}_\mu \, N^{\bar{\lambda}}{}_\nu \, ( g_{,\bar{\lambda}} \, g^{-1} )_{,\kappa} \, 
      -  N^\kappa{}_\nu \, N^{\bar{\lambda}}{}_\mu \, ( g_{,\bar{\lambda}} \, g^{-1})_{,\kappa}  = 0 \, .   
                           \label{g_N1_N2_eq_special}
\ee

In the case under consideration, we have $A_{\bar{\mu}} = 0$, so that the linear systems enforce 
$\bsy{\kappa}_{\bar{\mu}}  = \bsy{\lambda}_{\bar{\mu}} = 0$. The general binary Darboux transformation 
theorem then reduces to the following.

\begin{theorem}
Let real functions $N^\nu{}_\mu$ and  $N^{\bar{\nu}}{}_\mu$ satisfy
\be
   N^\mu{}_{\nu,\bar{\kappa}} = 0 \, , \quad
   N^\kappa{}_{[\mu} \, N^{\bar{\lambda}}{}_{\nu], \kappa} + N^{\bar{\lambda}}{}_\kappa \, N^\kappa{}_{[\mu,\nu]} = 0 
               \, , \quad
   N^{\bar{\kappa}}{}_{[\mu} \, N^{\bar{\lambda}}{}_{\nu],\bar{\kappa}} = 0  \, .   \label{gsdYM_Nijenhuis}
\ee
Let $r \times r$ matrices $Q, P,\bsy{\kappa}_\mu, \bsy{\lambda}_\mu$ be solutions of
\be
   &&  N^{\bar{\nu}}{}_\mu \, P_{,\bar{\nu}} - [\bsy{\kappa}_\mu , P] = N^\nu{}_\mu \, P P_{,\nu} \, , \quad
           N^{\bar{\nu}}{}_\mu \, Q_{,\bar{\nu}} + [\bsy{\lambda}_\mu , Q] = N^\nu{}_\mu \, Q_{,\nu} Q \, , 
                        \label{gsdYM_P,Q} \\
   && N^{\bar{\kappa}}{}_{[\mu} \, \bsy{\kappa}_{\nu],\bar{\kappa}} 
        - \frac{1}{2} [\bsy{\kappa}_\mu , \bsy{\kappa}_\nu] 
   = P \, ( N^\kappa{}_{[\mu} \, \bsy{\kappa}_{\nu], \kappa} + N^\kappa{}_{[\mu,\nu]} \, \bsy{\kappa}_\kappa ) \, , 
                      \nonumber \\
   && N^{\bar{\kappa}}{}_{[\mu} \, \bsy{\lambda}_{\nu],\bar{\kappa}} 
      + \frac{1}{2} [\bsy{\lambda}_\mu , \bsy{\lambda}_\nu] 
   = ( N^\kappa{}_{[\mu} \, \bsy{\lambda}_{\nu], \kappa} + N^\kappa{}_{[\mu,\nu]} \, \bsy{\lambda}_\kappa ) \, Q \, . 
                   \nonumber
\ee
Let $m \times r$, respectively $r \times m$ matrices $\theta$ and $\eta$ solve the linear systems
\bez
   &&    N^{\bar{\nu}}{}_\mu \, \theta_{,\bar{\nu}} - N^\nu{}_\mu \, \theta_{,\nu} \, Q 
                 = A_\mu \theta + \theta \, \bsy{\lambda}_\mu \, , \qquad
            N^{\bar{\nu}}{}_\mu \, \eta_{,\bar{\nu}} - N^\nu{}_\mu \, P \, \eta_{,\nu} 
                 = -  \eta \, A_\mu + \bsy{\kappa}_\mu \eta \, ,    \label{gsdYM_linsys}
\eez
with
\bez
          A_\mu = N^\nu{}_\mu \, \phi_{,\nu} \, , \quad \mbox{respectively} \quad
          A_\mu = N^{\bar{\nu}}{}_\mu \, g_{,\bar{\nu}} g^{-1} \, ,
\eez
where $\phi$ is a solution of (\ref{phi_N1_N2_eq_special}) and $g$ a solution of (\ref{g_N1_N2_eq_special}).
Furthermore, let $\Omega$ be an $r \times r$ matrix solution of
\bez
      && P \, \Omega - \Omega \, Q = \eta \, \theta \, , \\
      && N^{\bar{\nu}}{}_\mu \, \Omega_{,\bar{\nu}} = N^\nu{}_\mu \, ( \Omega_{,\nu} \, Q - P_{,\nu} \, \Omega 
             + \eta_{,\nu} \, \theta ) + \bsy{\kappa}_\mu \, \Omega + \Omega \, \bsy{\lambda}_\mu  \, .
\eez
Then
\be
           \phi' = \phi - \theta \, \Omega^{-1} \eta \, , \qquad
           g' = (I - \theta \, \Omega^{-1} P^{-1} \eta) \, g \, ,     \label{gsdYM_phi',g'}
\ee
are also solutions of (\ref{phi_N1_N2_eq_special}), respectively (\ref{g_N1_N2_eq_special}), on any open set 
of $\mathbb{R}^{2k}$, where $\Omega$ (and, in the second case, also $P$) is invertible.
\hfill $\Box$
\end{theorem}

\begin{example}
For vanishing seed, i.e., $A_\mu =0$, and $\bsy{\kappa}_\mu = \bsy{\lambda}_\mu=0$, let $\theta = \theta_0$ 
and $\eta = \eta_0$ be constant.  
Then, besides (\ref{gsdYM_Nijenhuis}), we only have to solve (\ref{gsdYM_P,Q}), which for $k=2$ are (in general 
non-autonomous) matrix versions of the Riemann (or Hopf, inviscid Burgers) equation. 
If $P = \mathrm{diag}(p_1,\ldots,p_r)$ and $Q= \mathrm{diag}(q_1,\ldots,q_r)$ are diagonal, non-constant, 
and have no eigenvalue in common (spectrum condition), 
then $\Omega$ is the Cauchy-like matrix with entries
\bez
         \Omega_{ij} = \frac{   (\eta_0 \,  \theta_0)_{ij}}{p_i - q_j}  \qquad i,j=1,\ldots,r \, ,
\eez
and (\ref{gsdYM_phi',g'}) typically provides us with breaking-wave solutions of (\ref{phi_N1_N2_eq_special}), 
respectively (\ref{g_N1_N2_eq_special}), which blow up in finite ``time" (here one of the coordinates). 
A special case appeared in Corollary~3.5 and Proposition~5.2 of \cite{CMHS15} (also see \cite{Zenc08sdYM}).
\hfill $\Box$
\end{example}

A subclass of integrable equations is obtained by choosing the matrix of $N_1$ as
$\mathrm{diag}(1,\ldots,1,0,\ldots,0)$, so that
\bez
                N^\mu{}_\nu = \delta^\mu_\nu \, ,
\eez
and all other components of the matrix of $N$ are zero. 
Then the conditions for the vanishing of the Fr\"olicher-Nijenhuis brackets are reduced to
\bez
           N^{\bar{\kappa}}{}_{[\mu,\nu]} = 0 \, , \qquad
           N^{\bar{\kappa}}{}_{[\mu} \, N^{\bar{\lambda}}{}_{\nu],\bar{\kappa}} = 0 \, ,
\eez
and (\ref{phi_N1_N2_eq_special}), (\ref{g_N1_N2_eq_special}) take the form
\be
            N^{\bar{\kappa}}{}_\mu \, \phi_{,\nu \bar{\kappa}} - N^{\bar{\kappa}}{}_\nu \, \phi_{,\mu \bar{\kappa}} 
         = [ \phi_{,\mu} , \phi_{,\nu}] \, ,     \label{phi_N1id_N2_eq_special}
\ee
respectively
\be
   N^{\bar{\kappa}}{}_\mu \, ( g_{,\bar{\kappa}} \, g^{-1})_{,\nu} - N^{\bar{\kappa}}{}_\nu \, ( g_{,\bar{\kappa}} \, g^{-1})_{,\mu} 
    = 0 \, .    \label{g_N1id_N2_eq_special}
\ee

\begin{example}
\label{ex:sdYM}
For $k=2$, i.e., in four dimensions, (\ref{phi_N1id_N2_eq_special}) and (\ref{g_N1id_N2_eq_special}) read
\bez
            N^{\bar{1}}{}_1 \, \phi_{,2 \bar{1}} + N^{\bar{2}}{}_1 \, \phi_{,2 \bar{2}}
            - N^{\bar{1}}{}_2 \, \phi_{,1\bar{1}}  - N^{\bar{2}}{}_2 \, \phi_{,1\bar{2}}  
         = [ \phi_{,1} , \phi_{,2}] \, ,
\eez
respectively 
\bez
       N^{\bar{1}}{}_1 \, ( g_{,\bar{1}} \, g^{-1} )_{,2} + N^{\bar{2}}{}_1 \, ( g_{,\bar{2}} \, g^{-1} )_{,2} 
       - N^{\bar{1}}{}_2 \, ( g_{,\bar{1}} \, g^{-1} )_{,1} - N^{\bar{2}}{}_2 \, ( g_{,\bar{2}} \, g^{-1} )_{,1} = 0 \, .
\eez
The conditions for the vanishing of the Fr\"olicher-Nijenhuis brackets are
\bez
      && N^{\bar{1}}{}_{1,2} = N^{\bar{1}}{}_{2,1} \, , \qquad 
       N^{\bar{2}}{}_{1,2} = N^{\bar{2}}{}_{2,1} \, , \\
      && N^{\bar{1}}{}_1 \, N^{\bar{1}}{}_{2,\bar{1}} +  N^{\bar{2}}{}_1 \, N^{\bar{1}}{}_{2,\bar{2}} 
            = N^{\bar{1}}{}_2 \, N^{\bar{1}}{}_{1,\bar{1}} +  N^{\bar{2}}{}_2 \, N^{\bar{1}}{}_{1,\bar{2}} \, , \\
     && N^{\bar{1}}{}_1 \, N^{\bar{2}}{}_{2,\bar{1}} +  N^{\bar{2}}{}_1 \, N^{\bar{2}}{}_{2,\bar{2}} 
            = N^{\bar{1}}{}_2 \, N^{\bar{2}}{}_{1,\bar{1}} +  N^{\bar{2}}{}_2 \, N^{\bar{2}}{}_{1,\bar{2}} \, .
\eez
Choosing $N^{\bar{1}}{}_1 = N^{\bar{2}}{}_2 =0$, the equations for $\phi$ and $g$ have the structure 
of familiar potential forms \cite{Chau83} of the 
self-dual Yang-Mills equation (on four-dimensional Euclidean space),
\bez
          N^{\bar{2}}{}_1 \, \phi_{,2 \bar{2}} - N^{\bar{1}}{}_2 \, \phi_{,1\bar{1}} = [ \phi_{,1} , \phi_{,2}] \, , \qquad
          N^{\bar{2}}{}_1 \, ( g_{,\bar{2}} \, g^{-1})_{,2} - N^{\bar{1}}{}_2 \, (g_{,\bar{1}} \, g^{-1})_{,1} = 0 \, ,
\eez
where now
\bez
           N^{\bar{1}}{}_{2,1} = N^{\bar{1}}{}_{2, \bar{2}} = 0 = N^{\bar{2}}{}_{1,2} = N^{\bar{2}}{}_{1,\bar{1}}  \, ,
\eez
so that $N^{\bar{1}}{}_2$ only depends on $x^2, \bar{x}^1$, and $N^{\bar{2}}{}_1$ only on $x^1, \bar{x}^2$.
Furthermore, setting $\bsy{\kappa}_\mu = \bsy{\lambda}_\mu = 0$, the matrices $P$ and $Q$ have to 
satisfy matrix versions of the Riemann equation,
\bez
     N^{\bar{2}}{}_1 \, P_{,\bar{2}} = P P_{,1} \, , \quad  N^{\bar{1}}{}_2 \, P_{,\bar{1}} = P P_{,2} \, , \quad
     N^{\bar{2}}{}_1 \, Q_{,\bar{2}} = Q_{,1} Q \, , \quad  N^{\bar{1}}{}_2 \, Q_{,\bar{1}} = Q_{,2} Q \, , 
\eez
and the linear systems take the form
\bez
   &&    N^{\bar{1}}{}_2 \, \theta_{,\bar{1}} - \theta_{,2} Q = A_2 \, \theta \, , \hspace{1.25cm}
       N^{\bar{2}}{}_1 \, \theta_{,\bar{2}} - \theta_{,1} Q = A_1 \, \theta \, , \\
   && N^{\bar{1}}{}_2 \, \eta_{,\bar{1}} - P \, \eta_{,2} = - \eta \, A_2  \, , \qquad
       N^{\bar{2}}{}_1 \, \eta_{,\bar{2}} - P \, \eta_{,1} = - \eta \,  A_1  \, .
\eez 
Hence, the linear systems depend on solutions of the above matrix Riemann equations.\footnote{This 
non-isospectrality is of fundamental importance in the case of integrable reductions of the Einstein 
(-Maxwell) vacum equations, see \cite{DMH13SIGMA}, in particular.}
\hfill $\Box$
\end{example}

\begin{example}
In the last case considered in the preceding example, we set $a:=N^{\bar{1}}{}_2$ and $b:=N^{\bar{2}}{}_1$, 
and restrict them to be constant. For vanishing seed, i.e., $A_\mu =0$, and constant matrices $P,Q$, the 
linear systems are solved by finite sums of the form 
\bez
      \theta = \sum_J \theta_J \, f_J (a x^2 I + \bar{x}^1 Q, b x^1 I + \bar{x}^2 Q) \, , \quad 
      \eta = \sum_J h_J(a x^2 I + \bar{x}^1 P, b x^1 I + \bar{x}^2 P) \, \eta_J \, ,
\eez
with constant $m \times r$ matrices $\theta_J$, constant $r \times m$ matrices $\eta_J$), and analytic functions $f_J,h_J$ of two variables. If $P$ and $Q$ are diagonal and satisfy the spectrum condition, 
then
\bez
         \Omega_{ij} = \frac{   (\eta \,  \theta)_{ij}}{p_i - q_j}  \qquad i,j=1,\ldots,r \, ,
\eez
and
\bez
                \phi' =  - \theta \, \Omega^{-1} \eta \, , \qquad
           g' = (I - \theta \, \Omega^{-1} P^{-1} \eta) \, g_0 \, ,
\eez
with an invertible constant $m \times m$ matrix $g_0$, solve the following potential versions of the self-dual 
Yang-Mills equations in four dimensions,
\bez
   b \, \phi_{,2 \bar{2}} - a \, \phi_{,1\bar{1}} = [ \phi_{,1} , \phi_{,2}] \, , \qquad
          b \, ( g_{,\bar{2}} \, g^{-1})_{,2} - a \, (g_{,\bar{1}} \, g^{-1})_{,1} = 0 \, .
\eez
\hfill $\Box$
\end{example}

In higher than four dimensions, i.e., for $k>2$, we get systems consisting of equations of the form obtained in Example~\ref{ex:sdYM}. 
Autonomous cases have been considered in \cite{Taka90icms,DMH00a,Gu08}.
% Different: \cite{Ward84,Suzu84,Sacl86,DKS92,Ivan+Popo93,BDK96}. 

\begin{example}
For $k=3$, (\ref{phi_N1id_N2_eq_special}) is the system of equations
\bez
    N^{\bar{\kappa}}{}_1 \, \phi_{,2 \bar{\kappa}} + N^{\bar{\kappa}}{}_2 \, \phi_{,2 \bar{\kappa}} = [ \phi_{,1} , \phi_{,2}] \, ,  \quad
    N^{\bar{\kappa}}{}_1 \, \phi_{,3 \bar{\kappa}} + N^{\bar{\kappa}}{}_3 \, \phi_{,1 \bar{\kappa}} = [ \phi_{,1} , \phi_{,3}] \, ,  \quad
    N^{\bar{\kappa}}{}_2 \, \phi_{,3 \bar{\kappa}} + N^{\bar{\kappa}}{}_3 \, \phi_{,2 \bar{\kappa}} = [ \phi_{,2} , \phi_{,3}] \, , 
\eez
and (\ref{g_N1id_N2_eq_special}) reads
\bez
     N^{\bar{\kappa}}{}_1 \, ( g_{,\bar{\kappa}} \, g^{-1})_{,2} = N^{\bar{\kappa}}{}_2 \, ( g_{,\bar{\kappa}} \, g^{-1})_{,1} \, , \quad
    N^{\bar{\kappa}}{}_1 \, ( g_{,\bar{\kappa}} \, g^{-1})_{,3} = N^{\bar{\kappa}}{}_3 \, ( g_{,\bar{\kappa}} \, g^{-1})_{,1} \, , \quad
    N^{\bar{\kappa}}{}_2 \, ( g_{,\bar{\kappa}} \, g^{-1})_{,3} = N^{\bar{\kappa}}{}_3 \, ( g_{,\bar{\kappa}} \, g^{-1})_{,2} \, .
\eez
\hfill $\Box$
\end{example}

Elementary Darboux transformations for autonomous generalized self-dual Yang-Mills equations appeared 
in \cite{Ma97,Gu08}. Our binary Darboux transformation is ``vectorial'' and generates multi-soliton solutions in a 
single step.

\subsection{Self-dual Yang-Mills hierarchy and generalizations}

Let $\cM$ be an infinite-dimensional manifold with local coordinates $x^{\mu_j}$, $\mu_j =1,\ldots,n$, $j=0,1,2,\ldots$.
A subscript of a Greek index thus specifies the corresponding subset of coordinates to which the Greek index refers.
As the matrices of $N_1$ and $N_2$, in these coordinates, we choose
\bez
          N_1 : \left( \begin{array}{ccccc} (N^{\mu_0}{}_{\nu_0}) & 0 & 0 &\cdots  \\
                                                                                                         0 & 0  & 0 & \cdots  \\
                                                                                        0 &  (N^{\mu_2}{}_{\nu_2}) & 0 & \cdots  \\ 
                                                                                                         0 & 0  & 0 & \cdots  \\
                                                                                                       \vdots & & \ddots &  \ddots  
             \end{array} \right)  \, , \quad
     N_2 : \left( \begin{array}{ccccc}   0& 0 & 0 &\cdots  \\
                                                                                                  (N^{\mu_1}{}_{\nu_0})  & 0  & 0 & \cdots  \\
                                                                                                  0 & 0  & 0 & \cdots  \\
                                                                                                  0 &  (N^{\mu_3}{}_{\nu_2}) & 0 & \cdots  \\                                                                                                    
                                                                                                  \vdots & & \ddots &  \ddots  
             \end{array} \right)  \, ,
\eez
We will assume that each block $(N^{\mu_{2p}}{}_{\nu_{2p}})$ is invertible. 
Then $[N_1,N_1]_{\mathrm{FN}} =0$ amounts to
\be
     &&   T((N^{\mu_{2p}}{}_{\nu_{2p}})) = 0 \, , \nonumber \\
      &&  N^{\kappa_{2q}}{}_{\nu_{2q}} \, N^{\lambda_{2p}}{}_{\mu_{2p}, \kappa_{2q}} 
         - N^{\lambda_{2p}}{}_{\kappa_{2p}} \, N^{\kappa_{2p}}{}_{\mu_{2p},\nu_{2q}} = 0 \quad \mbox{for}
        \; p \neq q \, ,  \nonumber \\
    &&  N^{\mu_{2p}}{}_{\nu_{2p}, \kappa_{2q+1}} = 0 \, .    \label{[N1,N1]=0}
\ee
Assuming also invertibility of $(N^{\mu_{2p+1}}{}_{\nu_{2p}})$, $[N_2,N_2]_{\mathrm{FN}} =0$ leads to the conditions
\bez
     &&  N^{\lambda_{2p+1}}{}_{\mu_{2p}, \kappa_{2p+1}} \, N^{\kappa_{2p+1}}{}_{\nu_{2p}} 
       - N^{\lambda_{2p+1}}{}_{\nu_{2p}, \kappa_{2p+1}} \, N^{\kappa_{2p+1}}{}_{\mu_{2p}} = 0 \, , \\
     &&   N^{\mu_{2p+1}}{}_{\nu_{2p}, \kappa_{2q+1}} = 0 \quad \mbox{for} \; p \neq q \, ,
\eez
and $[N_1,N_2]_{\mathrm{FN}} =0$ requires in addition
\bez
   &&   N^{\kappa_{2p}}{}_{[\mu_{2p}} \,  N^{\lambda_{2p+1}}{}_{\nu_{2p}], \kappa_{2p}} 
           +  N^{\lambda_{2p+1}}{}_{\kappa_{2p}} \,  N^{\kappa_{2p}}{}_{[\mu_{2p},\nu_{2p}]} = 0 \, , \\
   &&    N^{\kappa_{2p}}{}_{\mu_{2p}} \, N^{\lambda_{2q+1}}{}_{\nu_{2q}, \kappa_{2p}} 
         -  N^{\lambda_{2q+1}}{}_{\kappa_{2q}} \, N^{\kappa_{2q}}{}_{\nu_{2q},\mu_{2p}} = 0 \quad \mbox{for} \; p \neq q \, . 
\eez
 (\ref{phi_N1_N2_eq}) takes the form
\be
  &&    N^{\kappa_{2p+1}}{}_{\mu_{2p}} \, N^{\lambda_{2q}}{}_{\nu_{2q}} \, \phi_{,\lambda_{2q} \kappa_{2p+1}} - N^{\kappa_{2q+1}}{}_{\nu_{2q}} \, N^{\lambda_{2p}}{}_{\nu_{2p}} \, \phi_{,\lambda_{2p} \kappa_{2q+1}} = N^{\kappa_{2p}}{}_{\mu_{2p}} \, N^{\lambda_{2q}}{}_{\nu_{2q}} \, [\phi_{,\kappa_{2p}} , \phi_{,\lambda_{2q}}] \, , \nonumber \\
 && \hspace{5cm} p,q=0,1,2,\ldots \, ,      \label{gsdYMhier_phi}
\ee
using the last of equations (\ref{[N1,N1]=0}), and (\ref{g_N1_N2_eq}) reads
\bez
   &&   N^{\kappa_{2p}}{}_{\mu_{2p}} \, ( N^{\lambda_{2q+1}}{}_{\nu_{2q}} \, g_{,\lambda_{2q+1}} g^{-1} )_{,\kappa_{2p}}
   -  N^{\kappa_{2q}}{}_{\nu_{2q}} \, ( N^{\lambda_{2p+1}}{}_{\mu_{2p}} \, g_{,\lambda_{2p+1}} g^{-1} )_{,\kappa_{2q}}  \\
   && + N^{\kappa_{2p+1}}{}_{\lambda_{2p}} \, N^{\lambda_{2p}}{}_{\mu_{2p}, \nu_{2q}} \, g_{,\kappa_{2p+1}} g^{-1}
      - N^{\kappa_{2q+1}}{}_{\lambda_{2q}} \, N^{\lambda_{2q}}{}_{\nu_{2q}, \mu_{2p}} \, g_{,\kappa_{2q+1}} g^{-1} = 0 \, ,
  \quad p,q =0,1,2,\ldots \, ,  
\eez
which, by use of the required properties of the coefficients, reduces to
\be
    && N^{\kappa_{2p}}{}_{\mu_{2p}} \, N^{\lambda_{2q+1}}{}_{\nu_{2q}} \, ( g_{,\lambda_{2q+1}} g^{-1} )_{,\kappa_{2p}}
   -  N^{\kappa_{2q}}{}_{\nu_{2q}} \, N^{\lambda_{2p+1}}{}_{\mu_{2p}} \, ( g_{,\lambda_{2p+1}} g^{-1} )_{,\kappa_{2q}}  = 0 
          \nonumber \\
         && \hspace{5cm} p,q=0,1,2,\ldots \, ,                           \label{gsdYMhier_g}
\ee

The restriction to $p,q=0,1,$ leads back to the case considered in the preceding subsection. 
By construction, (\ref{gsdYMhier_phi}) and (\ref{gsdYMhier_g}) are integrable systems and they may be 
regarded as hierarchies associated with the generalized self-dual Yang-Mills equations.

\begin{remark}
Takasaki's self-dual Yang-Mills hierarchy \cite{Taka90} is based on a ``zero-curvature system" (see (2.1) in \cite{Taka90}),
of which an infinite-dimensional version of (\ref{0curv_system}) is a generalization. Nakamura's self-dual Yang-Mills hierarchy \cite{Naka88,Naka88red,Naka91} is included in Takasaki's \cite{Taka90}. Also see \cite{BLR82,ACT93} 
for related work.
\hfill $\Box$
\end{remark}

\section{Conclusions}
\label{sec:conclusions} 
We have shown that, in any dimension greater than one, there is a class of integrable systems, parametrized 
by two type $(1,1)$ tensor fields $N_1,N_2$, with vanishing Fr\"olicher-Nijenhuis brackets. Some familiar 
integrable systems are special cases: the chiral field equation, the self-dual Yang-Mills equation (in a certain gauge), higher-dimensional generalizations of the latter, and corresponding hierarchies. For these examples, the 
components of $N_1$ and $N_2$ are constant in some coordinate system, in which case the 
Fr\"olicher-Nijenhuis brackets are trivially zero. But vanishing of these brackets does not imply the 
existence of such coordinates, so that there are many non-autonomous generalizations of the 
aforementioned integrable systems. In the case of the chiral field equation, any non-constant holomorphic function 
determines a non-autonomous generalization. Most likely, further interesting integrable systems will 
emerge from the Fr\"olicher-Nijenhuis framework.

The familiar autonomous chiral field equation, self-dual Yang-Mills equations and their higher-dimensional 
generalizations, have been treated before in bi-differential calculus, using a bi-differential graded algebra 
based on a constant Grassmann algebra \cite{DMH00a}, which is not a differential-geometric setting. 
In the present work, we have shown that these systems actually show up in differential geometry. 
Moreover, the framework allows for substantial generalizations.

\vspace{.3cm}
\noindent
\textbf{Acknowledgments.} The author is grateful to Wen-Xiu Ma, Rusuo Ye and Yi Zhang for an invitation to present a 
series of lectures about the topic of this work at Zhejiang Normal University in Jinhua in April 2024, and for related 
discussions. 
He also thanks Xiaomin Chen for an invitation to talk about this at Beijing University of Technology in May 2024, and for inspiring discussions.

\end{document}